\documentclass{llncs}
\usepackage{makeidx}  
\usepackage{amsmath}
\usepackage{amssymb}
\usepackage{mathtools}
\usepackage{enumitem}
\usepackage{hyperref}
\usepackage{gastex}
\usepackage{pgfplots}
\usepackage{algorithmic}
\usepackage{algorithm}

\usepackage{cite}
\DeclarePairedDelimiter\abs{\lvert}{\rvert}%
\DeclareMathOperator{\height}{height}
\DeclareMathOperator{\cluster}{cluster}
\DeclareMathOperator{\tree}{tree}
\pgfplotsset{compat=1.13}
\begin{document}
\mainmatter          

\title{Implementation of the algorithm \\ for testing an automaton for synchronization \\ in linear expected time}
\titlerunning{Testing synchronization in linear expected time}

\author{Pavel Ageev\thanks{The author acknowledges support by the Competitiveness Enhancement Program of Ural Federal University.}}
\authorrunning{Pavel Ageev}
\tocauthor{Pavel Ageev}

\institute{Institute of Natural Sciences and Mathematics\\
Ural Federal University, Lenina 51, 620000 Yekaterinburg, Russia,\\
\email{pavel.ageev@urfu.ru}}

\maketitle
\begin{abstract}
Berlinkov has suggested an algorithm that, given a deterministic finite automaton $\mathcal{A}$, verifies whether or not $\mathcal{A}$ is synchronizing in linear (of the number of states and letters) expected time. We present a modification of Berlinkov's algorithm which we have implemented and tested. Our experiments show that the implementation outperforms the standard quadratic algorithm even for automata of modest size and allow us to give a statistically accurate approximation of the ratio of non-synchronizing automata amongst all automata with a given number of states.
\end{abstract}

\section{Background and Motivation}

A (deterministic finite) \emph{automaton} (DFA, for short) is a triple $\mathcal{A}=(Q, \Sigma, \delta)$, where $Q$ is a finite set of \emph{states}, $\Sigma$ stands for a finite \emph{alphabet} and $\delta\colon Q\times\Sigma \rightarrow Q$ is a \emph{transition function}. Let $\Sigma^*$ be the set of all words over $\Sigma$, including the empty word $\varepsilon$. Each word $w\in\Sigma^*$ acts on $Q$ via $\delta$: namely, for each state $q\in Q$, we let\footnote{Here and throughout expressions like $A := B$ mean that $A$ is defined to be $B$.}
\[
q.w:=\begin{cases}
q &\text{if}\ w=\varepsilon,\\
\delta(q.v,a) &\text{if}\ w=va\ \text{for some $v\in\Sigma^*$ and $a\in\Sigma$}.
\end{cases}
\]
This action extends to subsets of $Q$: for $D\subseteq Q$, we define $D.a:=\{q.a \mid q \in Q\}$.

An automaton $\mathcal{A}=(Q, \Sigma, \delta)$ is called \emph{synchronizing} if there exists a word $w \in \Sigma^*$ whose action is a constant function, i.e., $q.w = p.w$ for all $p,q \in Q$. Any word with such a property is called \emph{synchronizing}.

Automata serve to model real-world devices or protocols functioning in discrete mode. Often, such a device has to work within error-prone environment, and the property of being synchronizing may allow one to restore control over the device even if its current state has become unknown due to an error. We refer to the survey \cite{Volkov:2008} for a discussion and some illustrative examples. Therefore the natural question of how to determine whether or not a given DFA is synchronizing is of some importance. A well-known polynomial algorithm solving this question is based on the following observation.
\begin{proposition}
\label{prop:quadratic}
An automaton $\mathcal{A}=(Q, \Sigma, \delta)$ is synchronizing if and only if for each pair of states $p,q \in Q$, there exists a word $w\in\Sigma^*$ such that $p.w = q.w$.
\end{proposition}

It seems that this result first appeared in print in \v{C}ern\'y's pioneering paper \cite{Cerny:1964}, see Theorem~2 there. Independently
and somewhat earlier, it was obtained in Chung Laung Liu's PhD thesis~\cite{Liu:1963}, see Theorem~15 therein.

Given an automaton $\mathcal{A}=(Q, \Sigma, \delta)$, one can build the directed graph $\Gamma(\mathcal{A})$ whose vertices are all the ordered pairs of elements of $Q$ and whose edges are the pairs $((p_1,p_2),(q_1,q_2))$ such that there exists $a\in\Sigma$ with $p_i=\delta(q_i,a)$, $i=1,2$. The condition of Proposition~\ref{prop:quadratic} can be restated as follows: $\mathcal{A}$ is synchronizing if and only if each ordered pair of states is reachable in $\Gamma(\mathcal{A})$ by a directed path from a pair with equal entries. The latter condition can be checked by breadth-first search (BFS). If $|Q|=n$ and $|\Sigma|=k$, the graph $\Gamma(\mathcal{A})$ has $\tfrac{n(n+1)}2$ vertices and $\tfrac{kn(n+1)}2$ edges whence BFS in this graph requires $O(kn^2)$ time. Thus, we have a quadratic (in the number of states) algorithm to check the synchronizability of a given DFA. In the sequel, we refer to this algorithm as \emph{SynchSlow}. The algorithm is conceptually very simple but it may be inefficient already for automata with several thousand states.

For a clear definition of the random model of automata, which will be needed further, let for every state $q$ and a letter $a$, $\delta(q,a)$ is chosen uniformly at random from $Q$.

Berlinkov in his studies on synchronization of random automata~\cite{Berlinkov:preprint,berl} has suggested an algorithm that verifies whether or not a given DFA is synchronizing in linear expected time. Roughly speaking, this algorithm consists in verifying in a given automaton $\mathcal{A}$ a sequence of conditions each of which has two features:
\begin{itemize}
\item[(F1)] in every DFA with $n$ states, the condition can be checked in time $O(n)$;
\item[(F2)] the fraction of automata with $n$ states that do not satisfy the condition amongst all automata with $n$ states is $O(\frac1n)$.
 \end{itemize}
If $\mathcal{A}$ satisfies all these conditions, it is definitely synchronizing. If some of the conditions fails, the automaton may be synchronizing and may be not, and therefore, \emph{SynchSlow} should be called to get a definite answer. Thus, Berlinkov's algorithm may take quadratic time in the worst case. However, since the fraction of automata with $n$ states for which one needs invoking \emph{SynchSlow} is $O(\frac1n)$, the expected time that the described procedure spends when verifying the synchronizability of a given DFA with $n$ states will be $O(n)$.

In this paper we report our implementation of a (slight modification of) Berlinkov's algorithm and present some results of computational experiments. The experiments demonstrate that the implementation outperforms \emph{SynchSlow} for automata with more than 35 states. We use the experimental results to estimate the ratio of non-synchronizing automata amongst all automata with a given number of states.

\section{Description of the Algorithm}

Due to the space constraints, we do not reproduce separately the original version of Ber\-lin\-kov's algorithm as it may be found in \cite[Section~2]{berl} and (with more detail) in \cite[Section~4]{Berlinkov:preprint}\footnote{The latter source includes also some preliminary data of our early experiments.}. Instead we describe the implemented version of the algorithm. It basically follows the pattern of the original version, with a few modifications each of which will be explicitly specified. We have chosen to present the algorithm as a sequence of steps whose descriptions are interwoven with estimations of their running time and less formal comments.

We first consider the case of 2-letter alphabet. Thus, let $\mathcal{A}=(Q, \Sigma, \delta)$  be a DFA with $|\Sigma|=2$. We denote by $n$ the number of states in $\mathcal{A}$.

\smallskip

\textbf{Step 1.} Let $U\!G(\mathcal{A})$ stand for the \emph{underlying graph} of $\mathcal{A}$, that is, the directed graph with the vertex set $Q$ and the edge set consisting of all pairs $(q,p)\in Q\times Q$ such that $p=\delta(q,a)$ for some $a \in \Sigma$. The algorithm starts with finding the strongly connected components of $U\!G(\mathcal{A})$. This can be done in $O(n)$ time by Tarjan's algorithm \cite{Tarjan:72}. The reachability relation in $U\!G(\mathcal{A})$ induces a partial order on the set of strongly connected components. If $U\!G(\mathcal{A})$ has more than one strongly connected component which is minimal with respect to this order, then no state of $\mathcal{A}$ can be reached from every other state and the automaton is not synchronizing. The algorithm then returns \emph{``false''}.

For the rest of the description, we assume that the graph $U\!G(\mathcal{A})$ contains a unique minimal strongly connected component. Let $Q_0$ stand for the set of vertices of this component. Here the original version of the algorithm branches, depending on the size of $Q_0$: one proceeds if $|Q_0|\ge n/(4e^2)$ (where $e$ stands for the base of the natural logarithm); otherwise the algorithm \emph{SynchSlow} is called. We omit this check and proceed independently of the size of $Q_0$.

\begin{remark}
\label{rem:scc}
We would like to briefly discuss a subtlety that arises here. Let $\mathcal{A}_0$ denote the subautomaton of $\mathcal{A}$ obtained by restricting the transition function $\delta$ to the set $Q_0\times\Sigma$. Then it is easy to see that $\mathcal{A}$ is synchronizing if and only if so is $\mathcal{A}_0$. This fact may suggest that on this stage of the algorithm it is reasonable to start working with the smaller subautomaton $\mathcal{A}_0$ rather than the whole automaton $\mathcal{A}$.

The problem is that the initial automaton $\mathcal{A}$ is arbitrary whence, for each $a\in\Sigma$, the map $\delta(\_,a)\colon Q\to Q$ can be treated as a random map on the set $Q$. Some properties of such random maps are crucial for verifying that the conditions involved in Berlinkov's algorithm satisfy the feature (F2), that is, they hold in an overwhelming majority of automata. However, the restrictions of the maps $\delta(\_,a)$ to $Q_0$ cannot be treated as uniformly random since the subautomaton $\mathcal{A}_0$ is already not arbitrary (in particular, it is strongly connected).
\end{remark}

\textbf{Step 2.} For each letter $a\in\Sigma$, we denote by $U\!G(a)$ the underlying graph of the automaton
$(Q,\{a\},\delta|_{Q\times\{a\}})$. The graph $U\!G(a)$ consists of one or more (weakly) connected components, which will be referred to as \emph{clusters} hereafter. Every cluster consists of a single cycle (which may degenerate to a loop) and several trees whose roots lie on the cycle. For a state $q\in Q$, we denote by $\height(q)$ its \emph{height} with respect to a fixed letter $a\in\Sigma$, that is, the least non-negative integer $\ell$ such that $q.a^\ell=q.a^m$ for some $m>\ell$. Thus $q.a^{\height(q)}$ lies on some cycle of $U\!G(a)$ and serves as the root of the tree to which $q$ belongs. Fig.~\ref{fig:cluster} illustrates the notions just introduced.
\begin{figure}[th]
\begin{center}
\unitlength=0.75mm
\begin{picture}(110,55)
\node(A)(25,40){0}
\node(B)(49,26){0} \node(C)(37.5,5){0} \node(D)(12.5,5){0}
\node(E)(1,26){0} \drawedge(A,B){} \drawedge(B,C){}
\drawedge(C,D){} \drawedge(D,E){} \drawedge(E,A){}
\node(F)(70,26){1} \node(G)(88,37){2}
\node(I)(88,15){2} \drawedge(F,B){} \drawedge(G,F){}
\drawedge(I,F){} \node(J)(58,5){1}
\drawedge(J,C){} \node(K)(5,45){1} \node(L)(50,45){1}
\drawedge(K,A){} \drawedge(L,B){} \node(M)(70,45){2}
\drawedge(M,L){} \node(N)(88,56){3} \drawedge(N,M){}
\end{picture}
\end{center}\caption{A typical cluster with heights of its vertices shown}
\label{fig:cluster}
\end{figure}
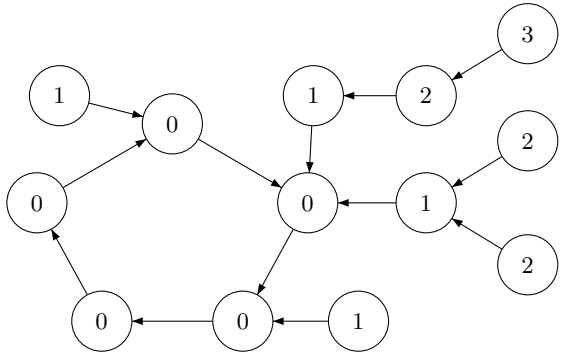

The algorithm proceeds by creating the data structure \emph{ClusterStructure} that, for each letter $a\in\Sigma$, contains:
\begin{itemize}
\item a list of indices of clusters;
\item for each index in the list, the size of the corresponding cluster, the length of its cycle, and a list of vertices that lie on the cycle (the latter list naturally indexes the trees of the cluster)
\item for each vertex $q\in Q$, the indices $\cluster(q)$ of the cluster and $\tree(q)$ of the tree to which $q$ belongs, and $\height(q)$.
\end{itemize}
\emph{ClusterStructure} can be built in $O(n)$ time, see~\cite[Lemma~9]{Berlinkov:preprint}.

Here the algorithm  branches, depending on the number of clusters: one proceeds if, for each $a\in\Sigma$, the number of clusters does not exceed $5\ln n$; otherwise \emph{SynchSlow} is called. By~\cite[Lemma~2]{Berlinkov:preprint}, the probability for \emph{SynchSlow} to be invoked at this step is $o(\frac{1}{n^4})$. Thus, for the rest of the description, we assume that
each graph $U\!G(a)$ has at most $5\ln n$ clusters.

\smallskip

\textbf{Step 3.} A 1-\emph{branch} is a subtree of one of the trees in a cluster of $U\!G(a)$ such that the root of this subtree has height~1. The height of a 1-branch is the maximum height of its vertices. For illustration, the cluster in Fig.~\ref{fig:cluster} has four 1-branches of which two have height~1 while two others have height~2 and~3.

Our algorithm checks if at least one letter $a\in\Sigma$ is such that $U\!G(a)$ has a unique 1-branch of maximum height (the \emph{tallest} 1-branch). It has been shown by Berlinkov~\cite[Theorem~3]{Berlinkov:tree} that the probability that a graph of the form $U\!G(a)$ (that is, the graph of a random map) has more than one 1-branch of maximum height is $O(\frac1{\sqrt{n}})$. If this happens for both letters in $\Sigma$, which event has the probability $O(\frac1n)$, \emph{SynchSlow} is called. Otherwise, we proceed, assuming that the tallest 1-branch (denoted $T$ in the sequel) exists for one of the letters in $\Sigma$; we denote this letter by $a_1$ and the other letter by $a_2$.

\smallskip

\textbf{Step 4.} On this step, which is specific for our modification, we check whether or not there exists a state $q\in Q_0\cap T$ such that the height of $q$ exceeds the height of any other 1-branch of $U\!G(a_1)$. (Recall that $Q_0$ denotes the set of vertices of the unique minimal strongly connected component of the graph $U\!G(\mathcal{A})$.) If this property fails, which happens with probability $O(\frac1n)$ by \cite[Theorem~6]{Berlinkov:preprint}, we call \emph{SynchSlow}. If it holds, we find the root $r$ of $T$ and the state $p$ that lies on the cycle of $\cluster(r)$ and is such that $p.a_1$ is the root of $\tree(r)$, see Fig.~\ref{fig:stable}. This step takes $O(n)$ time because we can calculate the height of each vertex end each 1-branch using single DFS. We also can verify if a certain state is in $q\in Q_0\cap T$ in constant time using precalculated indicator function for each of these sets.
\begin{figure}[th]
\begin{center}
\unitlength=0.75mm
\begin{picture}(110,55)
\node(A)(25,40){$p$}
\node(B)(49,26){} \node(C)(37.5,5){} \node(D)(12.5,5){}
\node(E)(1,26){} \drawedge(A,B){} \drawedge(B,C){}
\drawedge(C,D){} \drawedge(D,E){} \drawedge(E,A){}
\node(F)(70,26){} \node(G)(88,37){}
\node(I)(88,15){} \drawedge(F,B){} \drawedge(G,F){}
\drawedge(I,F){} \node(J)(58,5){}
\drawedge(J,C){} \node(K)(5,45){} \node(L)(50,45){$r$}
\drawedge(K,A){} \drawedge(L,B){} \node(M)(70,45){}
\drawedge(M,L){} \node(N)(88,56){} \drawedge(N,M){}
\end{picture}
\end{center}\caption{States forming a stable pair}
\label{fig:stable}
\end{figure}
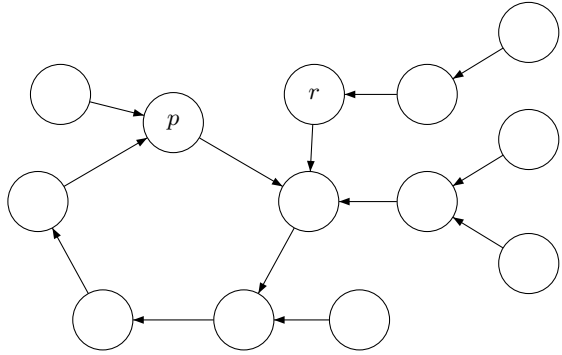

\textbf{Step 5.} It follows from \cite[Theorem~3]{Berlinkov:preprint} that the pair $\{p,r\}$ is \emph{stable} in the sense of~\cite{CKK:2002}. Recall that a pair of distinct states $\{q_1,q_2\}\in Q\times Q$ is called \emph{stable} in $\mathcal{A}=(Q, \Sigma, \delta)$ if, for each $w \in \Sigma^*$, there exists a word $v \in \Sigma^*$ such that $q_1.wv = q_2.wv$. It is known and easy to see that if $\{q_1,q_2\}$ is a stable pair and $w$ is an arbitrary word in $\Sigma^*$, then either $q_1.w = q_2.w$ or $\{q_1.w,q_2.w\}$ is a stable pair again. Using this, we `multiply' the stable pair $\{p,r\}$ as follows. First we construct 6 pairs $\{p_i,r_i\}:=\{p.a_2^i, r.a_2^i\}$, $i=1, 2, \dots, 6$.  Then for each $i=1,2,\dots,6$, consider the pairs $\{p_{ij},r_{ij}\}:=\{p_i.a_1^j, r_i.a_1^j\}$ where $j=1,2,\dots,\lceil n^{0.4}\rceil$. As discussed, all constructed pairs, except those whose entries coincide, are stable. We denote the set of stable pairs of the form $\{p_{ij},r_{ij}\}$ by $Z(a_2)$. Clearly, the set $Z(a_2)$ can be built in $O(n^{0.4})$ time. If $|Z(a_2)|<6$, we call \emph{SynchSlow}. The results of~\cite[Lemmas~7 and~8]{Berlinkov:preprint} ensure that this will happen with probability $O(\frac1n)$. If $|Z(a_2)|\ge6$, we take the 6 first pairs in $Z(a_2)$, and `multiply' them in the same fashion, that is, we act on each of these pairs by the words $a_2^j$ for $j=1,2,\dots,\lceil n^{0.4}\rceil$.  Then we select the stable pairs amongst the $6\lceil n^{0.4}\rceil$ pairs we get and denote the resulting set of stable pairs by $Z(a_1)$. Constructing $Z(a_1)$ also takes $O(n^{0.4})$ time.

Here the original version of Berlinkov's algorithm branches, depending on the sizes of the sets $Z(a_1)$ and $Z(a_2)$: one proceeds if each of these sets contains at least $\lceil n^{0.4} \rceil$ pairs; otherwise the algorithm \emph{SynchSlow} is called. In our modification we omit this check.

\smallskip

\textbf{Step 6.} Recall that we broke the symmetry of letters in Step~3. However, in the remaining steps of the algorithm, the difference between the letters $a_1$ and $a_2$ plays no role, and we use $a$ to denote any of this letter. A cluster of the graph $U\!G(a)$ is said to be \emph{large} if its size is greater than $n^{0.45}$. We denote by $L_a$ the set of all large clusters of the graph $U\!G(a)$ and consider the (undirected) graph $\Gamma_a$ with the vertex set $L_a$ whose edges are determined by the stable pairs in $Z(a)$ as follows: there is an edge between clusters $c,c'\in L$ whenever there exists a stable pair $\{p, q\} \in Z(a)$ such that $p \in c$ and $q\in c'$. Observe that by construction, $\Gamma_a$ may have loops and multiple edges. Our next step, which is specific for our modification, is to check whether or not the graph $\Gamma_a$ is connected for each $a\in\Sigma$. This can be done in sublinear time since $|L_a|\le 5\ln n$ and $|Z(a)|\le 6\lceil n^{0.4}\rceil$. By \cite[Lemma~3]{Berlinkov:preprint} the connectivity check fails with probability $O(\frac1n)$, and in this case \emph{SynchSlow} is called.

If the graph $\Gamma_a$ is connected, we proceed by computing the greatest common divisor $d$ of the lengths of cycles of the clusters in $L_a$. Using the Euclidean algorithm, one can find $d$ in $O(\ln^2n)$ time. If $d=1$ for each $a\in\Sigma$, we jump to the next step; if $d>1$ for some $a\in\Sigma$, an additional check is needed. Fix a spanning tree $\Theta$ of the graph $\Gamma_a$. (In the actual implementation, we build $\Theta$ when checking the connectivity of $\Gamma_a$.) Using $\Theta$, we label the vertices of $\Gamma_a$ by residues modulo $d$ as follows. The root of $\Theta$ gets label 0. Now suppose that some cluster $c\in L_a$ has already been labelled by $\ell(c)\in\{0,1,\dots,d-1\}$ while its child $c'$ in $\Theta$ has not yet got a label. We fix a stable pair $\{p,q\}\in Z(a)$ such that $p \in c$ and $q\in c'$ and define the label $\ell(c')$ from the congruence
\[
\ell(c')-\ell(c)\equiv \height(q) - \height(p)\pmod{d}\enspace.
\]
When the labelling process has been completed, we check whether the congruence
\begin{equation}
\label{eq:cong}
\ell(\cluster(s))-\ell(\cluster(t))\equiv \height(s) - \height(t)\pmod{d}
\end{equation}
holds true for every other stable pair $\{s,t\}\in Z(a)$ defining an edge in $\Gamma_a$. Since $|Z(a)|\le 6\lceil n^{0.4}\rceil$, this can be done in sublinear time. If follows from \cite[Lemma~4]{Berlinkov:preprint} that the probability that all congruences of the form~\eqref{eq:cong} simultaneously hold is $O(\frac1n)$, and in this case we call \emph{SynchSlow}. Otherwise we proceed
to the final step of the algorithm (that coincides with the final step of the original version).

\begin{remark}
\label{rem:logic}
This seems to be an appropriate place for a comment on the actual role of the conditions verified in Step 6 and on the overall logic of the algorithm. Let $\widehat{L_a}$ stand for the set of all states of $\mathcal{A}$ that belong to large clusters of the graph $U\!G(a)$.  It can be shown that if the graph $\Gamma_a$ is connected and either $d=1$ or $d>1$ but some of the congruences~\eqref{eq:cong} fails, then for each pair of states $p,q\in\widehat{L_a}$, there exists a word $w\in\Sigma^*$ such that $p.w = q.w$, see \cite[Lemmas~3 and~4]{Berlinkov:preprint}. In view of Proposition~\ref{prop:quadratic}, it remains to exhibit some additional conditions that hold with high probability and ensure the same conclusion for all pairs of different states with at least one entry lying beyond $\widehat{L_a}$. Exactly this is going to be done in Step~7.
\end{remark}

\textbf{Step 7.} Conditions to be described here involve both letters $a_1,a_2\in\Sigma$. In what follows, let $a$ denote any of these letters and let $b$ stand for the other one. Thus, each of the following items in fact represents \emph{two} conditions: one with $a=a_1$, $b=a_2$ and one with $a=a_2$, $b=a_1$.

\smallskip

\textbf{Step 7.1.} For each cycle of size $s>2$ in $U\!G(a)$, we find the number of its states belonging to the set $\widehat{L_b}$. If this number is at least $\lceil\frac{s}2\rceil$, i.e., the majority of the states are in $\widehat{L_b}$, we proceed; otherwise we call \emph{SynchSlow}. By \cite[Theorem~2, Case~1]{Berlinkov:preprint}, the probability of the latter event is $O(\frac1n)$, and clearly, the data collected in \emph{ClusterStructure} allow us to complete all verifications in this step in $O(n)$ time.

\smallskip

\textbf{Step 7.2.} For each cycle $C$ of size 2 in $U\!G(a)$ such that either $C\nsubseteq\widehat{L_b}$, we build the sets $C.b$ and $C.b^2$. We proceed if at least one of these sets is a singleton. If $|C.b|=|C.b^2|=2$, we proceed whenever $|C.b\cup C.b^2|=3$ and $C.b\subseteq\widehat{L_b}$ or $|C.b\cup C.b^2|=4$ and either $C.b\subseteq\widehat{L_b}$ or $C.b^2\subseteq\widehat{L_b}$. In all other cases we call \emph{SynchSlow}. By \cite[Remark~1]{Berlinkov:preprint}, the probability of invoking \emph{SynchSlow} at this step is $O(\frac1n)$, and again, all verifications we need clearly can be done in $O(n)$ time.

\smallskip

\textbf{Step 7.3.} For each cycle $C$ of $U\!G(a)$, we check whether or not the inclusions $C\subseteq\widehat{L_b}$ and $C.b\subseteq\widehat{L_a}$ hold true and store this information. This can be done in $O(n)$ time. Now we consider all pairs of different cycles $C,C'$ of $U\!G(a)$. Let $|C|=s$, $|C'|=s'$. We may assume that $s\ge s'$. If $s'\ge n^{0.45}$ (which implies that both $C$ and $C'$ belong to large clusters), we proceed to the next pair of cycles. If $s'=1$, i.e., the cycle $C'$ is a loop, we proceed to the next pair of cycles provided that either $C,C'\subseteq\widehat{L_b}$ or $C.b,C'.b\subseteq\widehat{L_a}$; otherwise we call \emph{SynchSlow}. Due to \cite[Theorem~2, Case~2]{Berlinkov:preprint}, the probability of the second alternative is $O(\frac1n)$.

If $1<s'\le n^{0.45}$, we compute the greatest common divisor $d$ of $s$ and $s'$ using the Euclidian algorithm that requires $O(\ln n)$ time. Let $C=\{p_0,p_1,\dots,p_{s-1}\}$ and $C'=\{q_0,q_1,\dots,q_{s'-1}\}$, where the states are listed in the order induced by the action of the letter $a$, that is, $p_i.a=p_{i+1\!\pmod{s}}$ for $i=0,1,\dots,s-1$, and $q_j.a=q_{j+1\!\pmod{s'}}$ for $j=0,1,\dots,s'-1$. Denote by $\mathbb{Z}_d$ the additive group of residues modulo $d$ and consider two subsets in this group:
\begin{gather*}
I:=\{i\mid\text{ for all } k\ge0,\ p_{i+kd\!\!\!\pmod{s}}\notin\widehat{L_b}\}\\
J:=\{j\mid\text{ for all } k\ge0,\ q_{j+kd\!\!\!\pmod{s'}}\notin\widehat{L_b}\}\enspace.
\end{gather*}
Then we check whether or not there is a `shift' $z\in\mathbb{Z}_d$ such that
\begin{equation}
\label{eq:shift}
\{z + i\mid i \in I\} \cup J = \mathbb{Z}_d\enspace.
\end{equation}
Both building sets $I$ and $J$ and searching for $z$ satisfying \eqref{eq:shift} can be done in $O(d^2)$ time, but since $d\le s'\le n^{0.45}$, we have that $d^2\le n^{0.9}$ so the time is sublinear in $n$. By \cite[Theorem~2, Case~2]{Berlinkov:preprint}, a shift $z$ verifying \eqref{eq:shift} exists with probability $O(\frac1n)$, and if this happens, we call \emph{SynchSlow}. Otherwise, we proceed to the next pair of cycles. Since the total number of pairs of cycles does not exceed $25 \ln{n}^2$, the total time spend on Step~7.3 is $O(n^{0.9} \ln{n}^2)=o(n)$. If \emph{SynchSlow} has not been invoked for any pair of cycles, our algorithm returns \emph{``true''}, that is, the automaton $\mathcal{A}$ is synchronizing. 

\smallskip

We have completed the description of the algorithm for DFAs with two input letters; let us call it the \emph{binary algorithm}. The extension to automata with $k>2$ input letters is fairly straightforward. Let $\mathcal{A}=(Q,\Sigma,\delta)$, where $\Sigma=\{a_1,a_2,\dots,a_k\}$, $k>2$. We first run the binary algorithm for the automaton $(Q,\{a_1,a_2\},\delta_{1,2})$, where $\delta_{1,2}$ is the restriction of $\delta$ to $Q\times\{a_1,a_2\}$. If the binary algorithm returns \emph{``true''} without calling \emph{SynchSlow}, then, clearly, $\mathcal{A}$ is synchronizing, and we return \emph{``true''} and stop. If the binary algorithm returns \emph{``false''} or if the necessity of calling \emph{SynchSlow} occurs, we apply the same procedure to the automaton $(Q,\{a_3,a_4\},\delta_{3,4})$, where $\delta_{3,4}$ is the restriction of $\delta$ to $Q\times\{a_3,a_4\}$, and so on. On the final step of the procedure, the binary algorithm is invoked for the automaton $(Q,\{a_{t-1},a_t\},\delta_{t-1, t})$, where $t=2\lfloor k/2\rfloor$ is the nearest even number less or equal than $k$. Again, if the binary algorithm returns \emph{``true''} without calling \emph{SynchSlow}, then $\mathcal{A}$ is synchronizing, and we return \emph{``true''} and stop. In all other cases, we call \emph{SynchSlow}. 

The probability that the binary algorithm returns \emph{``true''} for none of the automata $(Q,\{a_1,a_2\},\delta_{1,2}),\dots,(Q,\{a_{t-1},a_{t}\},\delta_{t-1,t})$ is $O(1/n^{\lfloor k/2 \rfloor})$, and the described procedure clearly takes $O(n)$ time for each fixed $k$.

We conclude this section with a comment on the nature of modifications made in our version of Berlinkov's algorithm. All the modifications consisted in 1)~omitting some of the conditions utilized in~\cite{Berlinkov:preprint} and 2)~checking a different condition instead. For instance, in Step~5 we omit the verification of whether each of the sets $Z(a_1)$ and $Z(a_2)$ contains at least $\lceil n^{0.4}\rceil$ pairs; instead, we check whether or not each of the graphs $\Gamma_a$ is connected in Step~6. The point is that, here and in all similar cases, it was the ``new'' condition that was implicitly used in~\cite{Berlinkov:preprint} while the role of the ``old'' condition was to ensure that the ``new'' one holds with high probability. (For instance, in the example just mentioned $Z(a_1)$ and $Z(a_2)$ serve as the edge sets for the graphs $\Gamma_{a_1}$ and respectively $\Gamma_{a_2}$; clearly, a graph in which the number of edges is much larger than the number of vertices is connected with high probability.) Therefore checking the ``new'' condition instead or the ``old'' one straightens the algorithm and decreases the probability of invoking \emph{SynchSlow}. As our experiments show, this has radically improved the algorithm's performance.

\section{Computational Experiments}

For brevity, we refer to the algorithm presented in the previous section as the \emph{main algorithm}. The algorithm was implemented in C++11. Compilation and assembly were made in Microsoft Visual Studio 2013 IDE (the compiler version MSVC 18.0.31101.0). All experiments were performed on a desktop PC with Intel Core i7-4770K (3.5GHz) CPU and 16Gb RAM. Source code can be found under the following link:\\ \url{https://github.com/birneAgeev/AutomataSynchronizationChecker}

In order to compare the main algorithm with \emph{SynchSlow}, we first experimented with its \emph{linear part} that returns \emph{fail} in the case of failure of any of the tests, instead of invoking \emph{SynchSlow}. Table~\ref{tab:avgtime} and Fig.~\ref{fig:avgtime} present the experimental results. For each combination of state/alphabet sizes, the average working time was computed from 1000 runs of the linear part on randomly generated automata. The graph in Fig.~\ref{fig:avgtime} confirms that the working time of our implementation of the linear part of the main algorithm indeed grows linearly with the number of states.
\begin{table}
\caption{The average working time (in seconds) of the linear part of the main algorithm for DFA with $n$ states and $k$ letters}
\centering
\begin{tabular}{l|c|c|c|c|c}
\hline\noalign{\smallskip}
&\quad $n=100$ \quad & \quad $n=1000$\quad  &\quad  $n=5000$\quad  &\quad  $n=10000$ \quad & \quad $n=100000$\\
\noalign{\smallskip}
\hline
\noalign{\smallskip}
  $k=2$ & 0.00013 & 0.00112 & 0.00554 & 0.0114 & 0.144\\
 $k=10$ & 0.00014 & 0.00113 & 0.00560 & 0.0115 & 0.159\\
$k=100$ \quad & 0.00014 & 0.00117 & 0.00564 & 0.0119 & 0.169\\
\hline
\end{tabular}
\label{tab:avgtime}
\end{table}

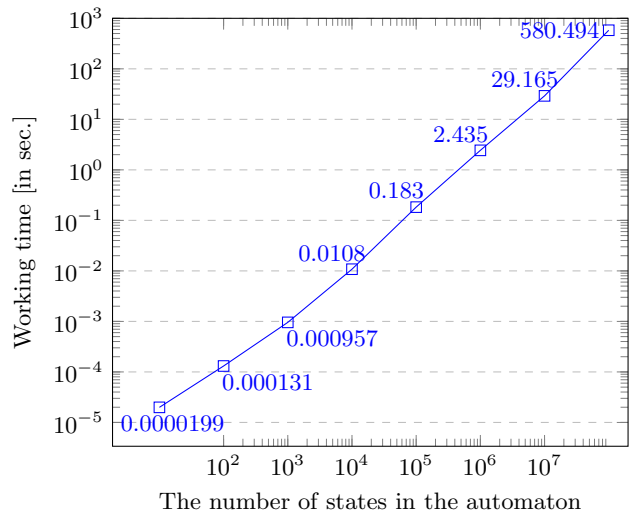
\begin{figure}[h]
\centering
\begin{tikzpicture}
\begin{axis}[
    xlabel={The number of states in the automaton},
    ylabel={Working time [in sec.]},
    xmax=200000000,
    ymax=1000.0,
    xtick={0,100,1000,10000,100000,1000000,10000000},
    ytick={0,0.00001,0.0001,0.001,0.01,0.1,1,10,100,1000},
    ymajorgrids=true,
    grid style=dashed,
    xmode=log,
    ymode=log
]

\addplot[
    color=blue,
    mark=square,
    visualization depends on=\thisrow{alignment} \as \alignment,
    nodes near coords, 
    point meta=explicit symbolic, 
    every node near coord/.style={anchor=\alignment}
    ]
    table [
     meta index=2 
     ] {
    x         y         label       alignment
    10        0.0000199 0.0000199   130
    100       0.000131  0.000131    160
    1000      0.000957  0.000957    160
    10000     0.0108    0.0108      -40
    100000    0.183     0.183       -40
    1000000   2.435     2.435       -40
    10000000  29.165    29.165      -40
    100000000 580.494   580.494     0
    };

\end{axis}
\end{tikzpicture}
\caption{The average working time of the linear part of the main algorithm for DFA with 2 input letters}
\label{fig:avgtime}
\end{figure}

We have used our experimental results to estimate the ratio of non-synchro\-niz\-ing automata amongst all automata with $n$ states and 2 input letters. According to Berlinkov's result~\cite[Theorem~1]{Berlinkov:preprint}, this ratio behaves as $\Theta(\frac{1}{n})$, in other words, the ratio of non-synchronizing automata amongst all automata with $n$ states and 2 input letters can be expressed as $\frac{c_n}n$ where the sequence $\{c_n\}$ tends to a non-zero limit as $n\to\infty$. Clearly, the ratio $\frac{f_n}n$ of failures of the linear part of the main algorithm gives an over-estimation of the ratio $\frac{c_n}n$. Hoeffding's inequality \cite{hoeffding} allows us to determine the number of runs sufficient to give statistically accurate bounds for the quantity $f_n$. Denote by $f_n(r)$ the number of automata among $r$ randomly chosen automata with $n$ states and 2 input letters for which the linear part of the main algorithm fails.

\begin{proposition}
\label{prop:estimation}
Let $\varepsilon>0$ and $0<p_0<1$. The number $t$ such that, with the probability at least $p_0$, the deviation of $f_n$ from the value $\frac{f_n(t n)}{t}$ is less than $\varepsilon$, can be computed from the inequality  $t\ge-\frac{n}{2\varepsilon^2}\ln{\frac{1-p_0}{2}}$.
\end{proposition}

\begin{proof}
We may consider the failure of the linear part of the main algorithm on a random automaton with $n$ states as a Bernoulli random variable with the probability $\frac{f_n}{n}$. Then Hoeffding's inequality applies, yielding
\begin{equation}
\label{eq:hoeffding}
P\big[(\frac{f_n}{n} - \gamma)t n \le f_n(t n) \le (\frac{f_n}{n} + \gamma)t n\big] \ge 1 - 2e^{-2 \gamma^2 tn}
\end{equation}
for every $\gamma>0$ and every integer $t$. Transform the expression for probability: 
\begin{multline*}
P\big[(\frac{f_n}{n} - \gamma)t n \le f_n(t n) \le (\frac{f_n}{n} + \gamma)t n\big] =\\ 
P\big[f_n - \gamma n \le \frac{f_n(t n)}{t} \le f_n + \gamma n\big] = P\big[\abs{\frac{f_n(t n)}{t} - f_n} \le \gamma n\big]\enspace.
\end{multline*} 
Now, we choose $\gamma=\frac{\varepsilon}n$ and require that $P\big[\abs{\frac{f_n(t n)}{t} - f_n} \le \varepsilon] \ge p_0$. Then we conclude from \eqref{eq:hoeffding} that $p_0 \le 1 - 2e^{\frac{-2\varepsilon^2t}{n}}$. After simple transformations, we get $t \geq -\frac{n}{2 \varepsilon^2} \ln{\frac{1-p_0}{2}}$.\qed
\end{proof}

Using Proposition~\ref{prop:estimation}, we see that, to calculate $f_n$ up to $\varepsilon=0.1$ with probability $p_0=0.99$, it suffices to run the linear part of the main algorithm $-\frac{n^2}{0.02} \ln{0.005} < 1060n^2$ times. The results of our calculations are collected in Table~\ref{tab:cn}. (For $n \geq 1000$, we used $\varepsilon=1$, which reduces the number of runs to $3n^2$, to make the calculations feasible.) As we mentioned above, the algorithm and all estimations are stated in the scope of the simple random model of an automaton, namely, for every state $q$ and letter $a$, $\delta(q, a)$ is chosen uniformly at random from $Q$. Nevertheless, other random models could be considered for a better understanding of limits on the use of the algorithm. We have also performed calculations of $f_n$ in the nonisomorphic model, in which automata are selected randomly from the set of all nonisomorphic automata with $n$ states and $k$ letters. For a generation of input data for the algorithm, the FAdo tool \cite{almeida2009fado} has been used.

\begin{table}
\caption{The estimation of $f_n$}
\centering
\begin{tabular}{c|c|c|c|c|c|c|c}
\hline\noalign{\smallskip}
$n$ \quad & \quad $5$ \quad & \quad $10$ \quad & \quad $20$ \quad& \quad$50$ \quad&\quad $100$ \quad&\quad $1000$ \quad& \quad$10000$\\
\noalign{\smallskip}
\hline
\noalign{\smallskip}
  $f_n$ in the simple uniform model \quad & 3.57 & 4.6 & 4.8 & 4.33 & 3.79 & 5.09 & 5.32\\
  $f_n$ in the nonisomorphic model \quad & 3.10 & 3.31 & 3.01 & 2.7 & 2.21 & --- & ---\\
\hline
\end{tabular}
\label{tab:cn}
\end{table}

Table \ref{tab:cn} shows that $f_n$ weakly depends on $n$ and seems to tend to a constant. Recall that $f_n$ is an overestimation for $c_n$, and thus, our results indicate that the ratio of of non-synchronizing automata amongst all automata with $n$ states and 2 input letters is upper-bounded by $\frac{5}{n}$ for the both models. This is consistent with theoretical results from \cite{berl,Berlinkov:preprint}. So, we can state a conjecture, that the similar asympthotic for the probability of being synchronizable is true in the nonisomorphic model. It is worth noting that the original algorithm described in \cite{Berlinkov:preprint} has the estimation of $f_n$ close to $n$ for all $n$ up to several thousand.

Estimation of the time complexity of the main algorithm is a nontrivial task because the quadratic algorithm \emph{SynchSlow} can be invoked, and for large $n$, e.g. $n=5000$, \emph{SynchSlow} requires too much time. Therefore we use our estimate of the value $f_n$ to find the running time of the main algorithm on $n$-state automata.

\begin{enumerate}
\item Estimate the constant $f_n$ as described above.
\item Calculate the total running time $t_{lin}$ of the linear part of the main algorithm on $n$ random automata with $n$ states.
\item Calculate the running time $t_{quad}$ of \emph{SynchSlow} on a single automaton with $n$ states (taking the average time of several runs on random automata).
\item Use $\dfrac{t_{lin} + f_nt_{quad}}{n + f_n}$ as an estimation for the average running time of the main algorithm on single automaton with $n$ states.
\end{enumerate}

From the data in Table \ref{tab:cn}, we may assume that $f_n \leq 6$ for $n \leq 10000$. Applying the above procedure, we get results presented in Table \ref{tab:maintime} and Fig.~\ref{fig:maintime}.

\begin{table}[t]
\caption{The average working time (in seconds) of the main algorithm for DFA with $n$ states and 2 input letters}
\centering
\begin{tabular}{c|c|c|c|c|c|c|c|c}
\hline\noalign{\smallskip}
$n$ &\quad 1000 \quad&\quad 2000 \quad&\quad 3000 \quad&\quad 4000 \quad&\quad 5000 \quad&\quad 7000 \quad&\quad 9000 \quad&\quad 10000\\
\noalign{\smallskip}
\hline
\noalign{\smallskip}
$\dfrac{t_{lin} + f_nt_{quad}}{n + f_n}$ \rule[-10pt]{0pt}{18pt} \quad& 0.002 & 0.005 & 0.008 & 0.010 & 0.014 & 0.021 & 0.027 & 0.031\\
\hline
\end{tabular}
\label{tab:maintime}
\end{table}

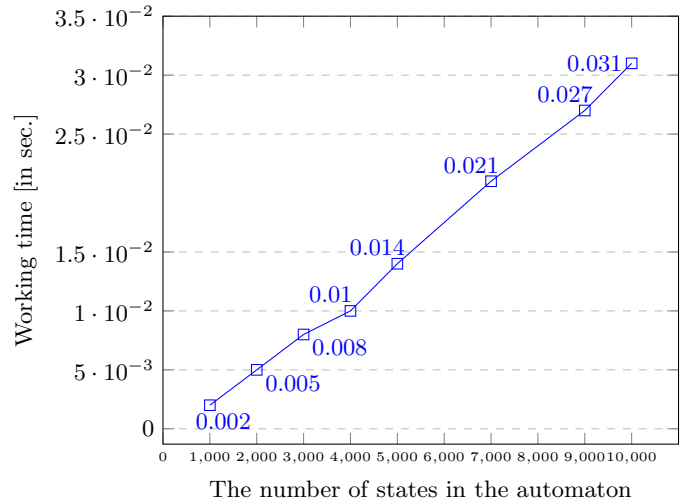
\begin{figure}
\centering
\begin{tikzpicture}
\begin{axis}[
    xlabel={The number of states in the automaton},
    ylabel={Working time [in sec.]},
    xmax=11000,
    ymax=0.035,
    xtick={0,1000,2000,3000,4000,5000,6000,7000,8000,9000,10000},
    ytick={0,0.005,0.01,0.015,0.2,0.025,0.03,0.035},
    ymajorgrids=true,
    grid style=dashed,
    scaled y ticks = false,
    scaled x ticks = false,
    every x tick label/.append style={font=\tiny}
]

\addplot[
    color=blue,
    mark=square,
    visualization depends on=\thisrow{alignment} \as \alignment,
    nodes near coords, 
    point meta=explicit symbolic, 
    every node near coord/.style={anchor=\alignment}
    ]
    table [
     meta index=2 
     ] {
    x         y     label    alignment
    1000      0.002 0.002    130
    2000      0.005 0.005    160
    3000      0.008 0.008    160
    4000      0.01  0.01     -40
    5000      0.014 0.014    -40
    7000      0.021 0.021    -40
    9000      0.027 0.027    -40
    10000     0.031 0.031     0
    };

\end{axis}
\end{tikzpicture}
\caption{The average working time of the main algorithm for DFA with 2 input letters}
\label{fig:maintime}
\end{figure}

Finally, we report about the comparison between the main algorithm and the quadratic algorithm \emph{SynchSlow}. Both algorithms were run $2650000$ times for all small $n$ in order to find minimal $n_0$ such that the average running time of \emph{SynchSlow} becomes greater than the average running time of the main algorithm for automata with at least $n_0$ states. The results of the comparison presented in Fig.~\ref{fig:compare} show that $n_0=31$.

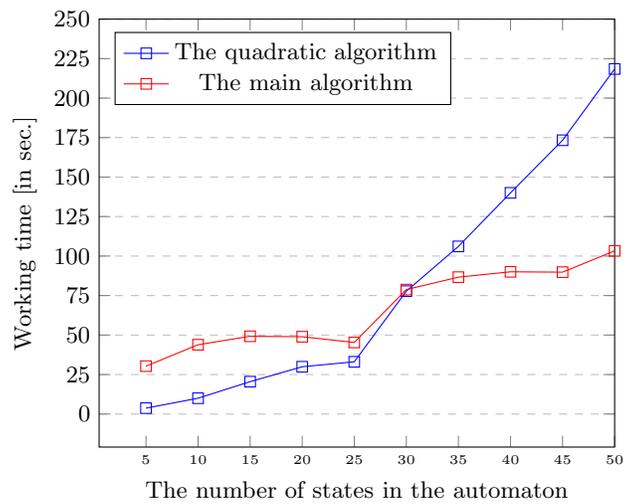
\begin{figure}
\centering
\begin{tikzpicture}
\begin{axis}[
    xlabel={The number of states in the automaton},
    ylabel={Working time [in sec.]},
    xmax=50,
    ymax=250,
    xtick={0,5,10,15,20,25,30,35,40,45,50},
    ytick={0,25,50,75,100,125,150,175,200,225,250},
    ymajorgrids=true,
    grid style=dashed,
    legend pos=north west,
    every x tick label/.append style={font=\tiny}
]

\addplot[
    color=blue,
    mark=square,
    ]
    coordinates {
        (5 , 3.726  )
		(10, 10.009 )
		(15, 20.496 )
		(20, 29.937 )
		(25, 33.075 )
		(30, 77.792 )
		(35, 106.104)
		(40, 140.015)
		(45, 173.314)
		(50, 218.424)
    };

\addplot[
    color=red,
    mark=square,
    ]
    coordinates {
        (5 , 30.321)
		(10, 43.879)
		(15, 49.222)
		(20, 48.935)
		(25, 45.319)
		(30, 78.779)
		(35, 86.649)
		(40, 90.032)
		(45, 89.806)
		(50, 103.28)
    };

\legend{The quadratic algorithm,The main algorithm}
\end{axis}
\end{tikzpicture}
\caption{The working time of the algorithms on $2650000$ automata}
\label{fig:compare}
\end{figure}

\noindent\textbf{Acknowledgements.} The author is grateful to his scientific advisor Professor Mikhail Volkov for helping with this article, to Mikhail Berlinkov for the explanation of his algorithm and to Mikhail Samoilenko for useful discussions.

\newpage

\bibliography{paper}

\end{document}